\documentclass[11pt,oneside,envcountsame,envcountsect]{llncs}
\usepackage{amsmath}
\usepackage{amssymb}
\usepackage{graphics}
\usepackage{enumerate}
\usepackage{gastex} 
\usepackage{xcolor} 

\usepackage[margin=2.02cm]{geometry}
\usepackage{graphicx}
\usepackage{xspace}
\usepackage[algoruled,vlined,english,linesnumbered]{algorithm2e}
\usepackage{algorithmic}
\usepackage[small,hang]{caption2}

\renewcommand{\check}[1]{{\bf #1}}

\title{Counting in One-Hop Beeping Networks}
\author{A. Casteigts, Y. M\'etivier,
  J.M. Robson and A. Zemmari} \institute{Universit\'e de Bordeaux - 
Bordeaux INP\\
  LaBRI UMR CNRS 5800\\ 351 cours de la Lib\'eration, 33405 Talence,
  France\\ \{acasteig, metivier, robson, zemmari\}@labri.fr } \date{ }
\begin{document}
\maketitle
\newtheorem{fact}{Fact}[section]
\begin{abstract}

We consider networks of processes which interact with beeps. In the basic model defined by Cornejo and Kuhn~\cite{Cornejo10}, which we refer to as the $BL$ variant, processes can choose in each round either to beep or to listen. Those who beep are unable to detect simultaneous beeps. Those who listen can only distinguish between silence and the presence of at least one beep. Beeping models are weak in essence and even simple tasks may become difficult or unfeasible with them.

In this paper, we address the problem of computing how many participants there are in a one-hop network: 
the {\em counting} problem. We first observe that no algorithm can compute this number with certainty in $BL$, whether the algorithm be deterministic or even randomised (Las Vegas). We thus consider the stronger variant where beeping nodes are able to detect simultaneous beeps, referred to as $B_{cd}L$ (for {\em collision detection}). We prove that at least $n$  rounds are necessary
 in $B_{cd}L$, and we present an algorithm whose running time is $O(n)$ rounds with high probability. Further experimental results show that its expected running time is less than $10n$. Finally, we discuss how this algorithm can be adapted in other beeping models. In particular, we show that it can be emulated in $BL$, at the cost of a logarithmic slowdown and of trading its Las Vegas nature (result certain, time uncertain) against Monte Carlo (time certain, result uncertain). 
\end{abstract}
{\bf keywords:} Beeping model, Size computation, Counting problem, Las Vegas algorithms.

\section{Introduction}

Distributed algorithms are concerned with assumptions relating to various aspects like the structure of the network (e.g. trees, rings, planar graphs, complete graphs) or knowledge available to the nodes (e.g. a bound on the network size, unique identifiers, or port numbering). Another important aspect is the size of messages, which may range from unbounded (the ${\cal LOCAL}$ model) to logarithmic size (${\cal CONGEST}$ model), to constant size (e.g. simple bits)~\cite{Peleg}.

A natural approach in distributed computing is to reduce the assumptions as much as possible, in order to make positive results more general. Hence, when a problem is solved in some strong model, one naturally strives to solve it in a weaker model. In a recent series of works~\cite{Cornejo10,Schneider10,Afek13,HuangM13,Scott13,GN15}, the community has started to explore new models of communications that are even weaker than constant size messages in anonymous networks, namely {\em beeping models}.

In these models, the only communication capabilities offered to the nodes are to {\em beep} or to {\em listen} for the beeps of others. Several variants exist. In~\cite{Cornejo10}, a node that beeps is unable to detect whether other nodes have beeped at the same time. When listening, it can distinguish between silence or the presence of at least one beep, but it cannot distinguish between one and several beeps. In Section~6 of~\cite{Afek13}, a stronger variant is considered where beeping nodes can detect whether other nodes are beeping simultaneously (sender side collision detection). In~\cite{Schneider10} and Section~4 of~\cite{Afek13}, yet another variant is considered where the nodes can tell the difference between silence, one beep, and several beeps.
 The ability to detect internal collision is denoted by $B_{cd}$ ($B$ otherwise) and that of detecting peripheral collisions is denoted by $L_{cd}$ ($L$ otherwise).
The various models in literature can be reformulated in these terms. The basic model introduced by Cornejo and Kuhn in~\cite{Cornejo10} is $BL$; the model considered by Afek et al. in~\cite{Afek13} (Section~6) and Scott et al. in~\cite{Scott13} is $B_{cd}L$; and the model considered in~\cite{Schneider10} and in Section~4 of~\cite{Afek13} is $BL_{cd}$. To the best of our knowledge, $B_{cd}L_{cd}$ has only been considered in~\cite{CMRZ16}.

Studying weak models of computation is interesting in its own right. In addition, beeping models turn out to be relevant to model real-world applications or phenomena. For instance, they reflect the features of a network at the lowest level (physical layer), where a node can probe or emit signals, with or without collision detection. At a higher level of abstraction, beeping models also reflect some communication patterns in biology, such as {\em lateral inhibition} among neighboring cells~\cite{Collier96}.

\paragraph{Contributions.}
In this paper, we consider the counting problem in which nodes must determine
the size of the network.
We focus on the case that the communication graph is complete ({\em one-hop} networks), implying that each node can hear and be heard by all the others. 
As it turns out, even this version of the problem cannot be solved deterministically, due to the inherent lack of symmetry-breaking mechanisms in the beeping model. We thus consider randomised algorithms and start by providing a {\it Las Vegas} (LV) algorithm -- correct result but uncertain time, though finite with probability $1$ -- for solving the 
counting problem 
 in $B_{cd}L$. The expected running time of this algorithm is linear in the number of nodes (with high probability). On the negative side, we show that $n$
is a lower bound, which makes the algorithm optimal up to a constant factor (experimentally estimated to less than $10$). Unfortunately, we also observe that no LV algorithm exists in the weaker $BL$ model, leading us to consider a {\it Monte Carlo} (MC) variant in $BL$ where the result is correct only with some threshold probability. This algorithm relies on a technique developed in~\cite{CMRZ16}, which enables the emulation of $B_{cd}$ beeps with $B$ beeps, at the cost of a logarithmic slowdown. We also discuss how to adapt the algorithm in the stronger $B_{cd}L_{cd}$ model, this time with a (constant factor) gain in time complexity.

\paragraph{Related Work:}
As explained by Chlebus \cite{C01},
detecting  a collision in a radio network is to be able to distinguish
between $0$ messages and at least $2$ messages while detecting a collision
in the beeping model is to be able to distinguish between $1$ message
and at least $2$ messages. Thus results on the counting problem  in the 
context of radio networks  cannot be applied directly in the context of 
beeping models. Another problem which has been well-studied, in this context,
is the $k$-Selection problem (also known as all-broadcast); as explained 
by Anta and Mosteiro \cite{AM10} (see also \cite{K05,BKKK16}): 
it is solved ``when an unknown size-$k$
subset of $n$ network nodes have been able to access a unique shared channel
 of communication, each of them at least once''. Finally,
selection problems, in general,
differ from the counting problem studied in this paper 
\begin{enumerate}
\item
by the fact that
the collision detection does  not have the same meaning in the context of radio 
networks and in the context of the beeping models, and 
\item by the fact that in our case we look for the exact number of nodes
and the $k$-selection problem asks that nodes access the unique shared channel
at least once.
\end{enumerate}

Afek et al. \cite{Afek13}, 
 from considerations concerning the development of
certain cells, studied the MIS problem in the discrete beeping model $BL$
as presented in \cite{Cornejo10}. They consider, in particular,
the wake-on-beep model (sleeping nodes wake up upon receiving a beep) and
sender-side collision detection $B_{cd}L$: they give a $O(\log^2 n)$ rounds 
MIS algorithm. After this work, 
Scott et al. \cite{Scott13} present in the model $B_{cd}L$ 
a randomised algorithm
with feedback mechanism whose expected time to compute a MIS is
$O(\log n)$. 

More generally, Navlakha and Bar-Joseph present in \cite{NB15}
a general survey on 
similarities and differences between distributed computations in biological
and computational systems and, in this framework, the importance
of the beeping model.

In \cite{Cornejo10}, Cornejo and Kuhn study the interval colouring problem:
an interval colouring assigns to each node an interval (contiguous fraction)
of resources such that neighbouring nodes do not share resources
(it is a variant of graph colouring). They assume that each node
knows its degree and an upper bound of the maximum degree $\Delta$
of the graph. They present in the  beeping
model $BL$ a probabilistic  algorithm which never stops
and stabilises with a correct  $O(\Delta)$-interval coloring  in
$O(\log n)$ periods with high probability, where:  $n$ is 
the size of the graph, $\Delta$ its maximum degree and a period is $Q$ time
slots with $Q\geq \Delta$, thus it stabilises in $O(Q \log n)$ slots.

Emek and Wattenhofer introduce in \cite{EmekW13}
a model for distributed computations which resembles the beeping model:
networked finite state machines (nFSM for short). This model enables the sending
of the same message to all neighbours of a node; however it is asynchronous,
the states of nodes belong to a finite set, the degree of nodes is bounded
and the set of messages is also finite. In the nFSM model they give a
$2$-MIS algorithm for graphs of size $n$ using a set of messages of size $3$
with a time complexity equal to $O({\log n}^2).$ \check

Concerning the counting problem, in the context of the classical message passing
model, counting Monte Carlo 
algorithms for anonymous rings are presented in \cite{IR90,MRZTCS15}. The  time complexity is $O(n)$.
It is also investigated in dynamic networks \cite{KLO10,LBCB13,MCS13}. In \cite{KLO10}
nodes have unique identifiers, the size of messages is $O(\log n)$. Nodes
know the size of the graph  in $O(n^2)$ rounds. Networks are dynamic and anonymous
in \cite{LBCB13}; there is a leader in the network and the termination of the
counting algorithm is detected by a heuristic. In \cite{MCS13} communication is asynchronous
message passing, nodes have no identifiers and there exists a leader; counting algorithm
which is presented produces a correct stabilizing solution that do not guarantee termination.
\paragraph{Summary.}
The paper is organised as follows. In Section~\ref{sec:definitions} 
we provide  definitions and basic observations regarding the beeping model. Section~\ref{sec:one_hop} presents our main results on the counting problem, including the lower bound, our optimal algorithm, and its time complexity analysis. Finally, Section~\ref{sec:adaptation} discusses the adaptations of our algorithm in other variants of the beeping model, including its emulation as a Monte Carlo algorithm in the weakest $BL$ variant. 

\section{Network Model and Definitions}
\label{sec:definitions}

We consider a wireless network model and
we follow definitions given in~\cite{Afek13} and~\cite{Cornejo10}.
The network is anonymous: unique identifiers are not available to
distinguish the processes. 
Communications are
synchronous and encoded by a graph $G=(V,E)$ where the nodes
$V$ represent processes and the edges $E$ represent pairs of processes
that can hear each other. Since we focus here on one-hop networks, $G$ is a complete graph.
Time is divided into discrete synchronised time intervals called {\em slots} (following the usual terminology in wireless networks). 
All processes wake up and start computation at the same slot.
In each slot, all processors act in parallel and either beep or listen. In addition, processors can perform an unrestricted amount of local computation in-between two slots (in effect, our algorithms require little computation).

In this paper, we consider several variants of beeping models:
\begin{itemize}
\item if a process beeps, there are two cases:
\begin{enumerate}
\item it cannot know whether another process beeps simultaneously, 
this case  is denoted by $B$;
\item it can distinguish whether it beeped alone or if at least one neighbour
beeped concurrently, it is an internal collision; 
this case is called sender side collision detection,
and it is denoted in this paper $B_{cd}$;
\end{enumerate}
\item if a process listens, there are also two cases:
\begin{enumerate}
\item
it can distinguish between silence or the presence of at least one beep, this
model is denoted $L$;
\item
  it can distinguish between silence or the presence of one beep or
   the presence of at least two beeps; in this case it is a peripheral
collision, this model
is denoted $L_{cd}$  in this paper.
\end{enumerate}
\end{itemize}

Finally, a beeping model is defined by choosing between $B$ or $B_{cd}$ and
between $L$ and $L_{cd}$. 
\begin{remark}\label{listen}
In general, nodes are active or passive. When they are active
they beep or listen; in the description of algorithms we say explicitely when 
a node beeps meaning that a non beeping active node listens.
\end{remark}

The time complexity, also called {\em slot complexity}, is the maximum number of slots needed until every node has completed its computation. Our algorithms are typically structured into {\em phases}, each of which corresponds to a small (constant or logarithmic) number of slots. In the algorithm, we specify which one is the current slot by means of a {\tt switch} instruction with as many {\tt case} statements as there are slots in the phase. Phases repeat until some condition holds for termination.

\begin{remark}
An algorithm given in a beeping model induces an algorithm in the (synchronous)
message passing model. Thus, given a problem,
any lower bound on the round complexity
in the message passing model also holds for slot complexity in the beeping model.
\end{remark}

In this paper, results on graphs
having $n$ nodes are expressed with high probability (w.h.p. for short),
meaning with probability $1-o(n^{-1})$.
We write $\log n$ for the binary logarithm of $n$.

\paragraph{Distributed Randomised Algorithm:}
A randomised (or probabilistic) algorithm is an algorithm which makes choices
 based on given probability distributions. 
A {\em distributed} randomised algorithm is a collection of local randomised
algorithms (in our case, all identical). 

A {\em Las Vegas} algorithm  is a randomised algorithm which terminates
with probability one, and always produces
 a correct result.
A {\em Monte Carlo} algorithm is a randomised algorithm which terminates deterministically, but whose result may be incorrect with a certain
probability. Intuitively, Las Vegas algorithms have uncertain execution time but certain result, while Monte Carlo does the reverse.
Classical considerations on symmetry breaking in anonymous 
beeping networks (see for instance Lemma~4.1 in~\cite{Afek13}), imply that:
\begin{remark}
There is no Las Vegas (and a fortiori no deterministic) algorithm in $BL$ which allows a node
to distinguish between an execution where it is isolated and one where it has exactly one neighbour.
\end{remark}

From this remark we deduce that there is no Las Vegas counting algorithm
in $BL$, which advocates the use of stronger models. In what follows, we consider the $B_{cd}L$ model, where a beeping node can detect simultaneous beeps by others. We give a Las Vegas algorithm in this model, which is then turned into a Monte Carlo algorithm in $BL$ using emulation techniques.

\paragraph{The counting problem:}
We say that an algorithm solves the counting problem if
and only if 
 by the end of its execution, {\em every} node knows the size of the network.

\section{One-Hop Network Size Computation in $B_{cd}L$}

In the $B_{cd}L$ model, a node that beeps can detect simultaneous beeps by others, while a node that listens cannot distinguish between one and several beeps. Hence, only beeping nodes can detect collisions, which proves sufficient to solve the counting problem. In this section, we first prove that at least $n$ 
slots are 
required to learn the number of nodes $n$; 
then we present a Las Vegas algorithm which takes $O(n)$ slots 
to terminate (w.h.p.). Simulation results refine this complexity 
to less than $10n$ slots ({\it i.e.} less than $3.32$ phases of $3$ slots each).

\label{sec:one_hop}
\subsection{A Lower Bound }
The following lemma establishes that $n$ slots are required for computing $n$ in a complete graph.

\begin{lemma}\label{lowerbound}
Any Las Vegas beeping algorithm that counts the number of nodes in a 
complete graph $K_n$ needs $n$ slots.
\end{lemma}
\begin{proof}[By contradiction]
Let $\cal A$ be such an algorithm and let $\cal E_{A}$ be an execution of $\cal A$ that terminates in less than $n$ slots in the complete graph $K_n$. Then it holds that at least one node, say $v$, never beeped alone. Let $\cal E_{A}'$ be another execution of $\cal A$, but this time in the complete graph $K_{n+1}$ composed of the same nodes as before, plus $v'$. Let all the nodes behave as they did over $\cal E_{A}$ and let $v'$ act exactly like $v$. Since $v$ never beeped alone in $\cal E_{A}$, it does  the same in $\cal E_A'$ and so does $v'$ as well, making both executions indistinguishable (two beeps are indistinguishable from three). Hence, the nodes in $\cal E_A'$ terminate as in $\cal E_A$, having counted $n$ instead of $n+1$, which is a contradiction.
\qed
\end{proof}
\begin{remark}
The bound holds even when collision detection is also available for listening nodes (thus, in $B_{cd}L_{cd}$ and a fortiori in $BL_{cd}$), due to the same argument that two beeps are indistinguishable from three on the listener side.
\end{remark}
\subsection{An Algorithm for 
Computing One-Hop Network Size  in  $B_{cd}L$}
We now propose a beeping algorithm that solves the counting problem in one-hop networks in $B_{cd}L$. Then we characterise its slot complexity both analytically and experimentally.

\subsubsection{Informal description:} The basic idea is to have nodes beep alone as often as possible, so that they can detect it and become counted. This is achieved by means of three slots: In slot $1$, the non-counted nodes beep with some probability. If they beeped alone, they re-beep in slot $2$ to inform the other nodes and they can be counted. In slot $3$, all non-counted nodes beep. If no one beeped in the third slot, then it means that every node is counted and the algorithm can terminate globally. The key ingredient is to use an adaptive probability $1/k$ to increase the chances to beep alone in the first slot, which is done by increasing the probability (decreasing $k$) whenever no one beeped, and decrease it (increasing $k$) whenever a collision is detected. As shown in the analysis section further below, this technique leads to a linear running time in the number of nodes.

The details are given on Algorithm~\ref{algo:size}. It describes the phase that repeats until termination is detected. Each {\tt case} statement corresponds to a slot, and an extra {\tt case} statement groups together various computation to be performed at the end of the phase: updating the count variable (line~\ref{line:count}), detecting termination (line~\ref{line:termination}), adjusting the probability (lines~\ref{line:proba1} to~\ref{line:proba2}).

\begin{algorithm}[h]
\label{algo:size}
\caption{Computing the size of a one-hop network in $B_{cd}L$.}  
$Boolean\ counted \gets false$ \\
$Boolean\ terminated \gets false$ \\
$Integer\ k \gets 2$ \\
$Integer\ size \gets 1$\\ 
\SetKwRepeat{Repeat}{repeat}{until}%
\Repeat{terminated = true}{
  \Switch{slot}{ 
    \uCase{1}{ 
      \If{$counted = false$}{ 
        beep with probability $\frac 1 k$ 
        \label{line:beep}
        \tcp*{contends}
      } 
    }
    \uCase{2}{
      \If{I am the only one that beeped in slot $1$}{
        \label{line:test}
        beep \tcp*{wins the contest}
        $counted \gets true$
      }
    }
    \uCase{3}{
      \If{ $counted = false$ }{
        beep \tcp*{will contend in the next phase}
      }
    }
\SetKwSwitch{Switch}{Case}{}{}{}{end of phase:}{}{}{}%
    \uCase{}{
      \If{a neighbor beeped in slot 2}{
        $size \gets size + 1$
        \label{line:count}
      }
      \If{no one beeped in slot 3}{
        $terminate \gets true$
        \label{line:termination}
      }
      \eIf {no one beeped in slot 1}{
        \label{line:proba1}
        \If{$(k>2)$}{
          $k \gets k-1$
        }
      }{\If{there was a collision in slot 1}{
        $k \gets k+1$
        \label{line:proba2}
      }
      }
    }
  }
}
\SetKwSwitch{Switch}{Case}{Other}{switch}{do}{case}{otherwise}{}{}
\SetKwRepeat{Repeat}{repeat}{until}%
\end{algorithm}

\newpage
\subsubsection{Analysis of the Algorithm:}
For the sake of analysis, it is easier to think in terms of phases rather than slots. The algorithm progresses every time a node beeps alone in the first slot of a phase. The probability that this happens in a phase depends upon the value of $k$ and the number of nodes $n'$ still contending in this phase. Intuitively, if $k$ is too large, then it is likely that no node will beep in the first slot, while if $k$ is too small, they will likely be several to do so. In fact, the probability of success is maximum when $k=n'$, which is what the algorithm attempts to maintain. The following fact is important:

\begin{fact}
Uncounted nodes all have the same $k$.
\end{fact}

In the analysis below, we distinguish between the case that $k$ is within ``good bounds'' and the case that it is not. We call a {\em bad} phase one in which $k$ drifts out of these bounds or keeps drifting away (if it is already out). More formally, let $I= [n',3n']$. A phase $\phi$ is a \emph{bad phase} if:
\begin{itemize}
\item[\emph{(i)}] at the start of $\phi$, $k \le n'$ or $k \ge 3n'$, and
\item[\emph{(ii)}] the phase moves $k$ farther away from $I$ (or $k$ remains unchanged at 2).
\end{itemize}
We prove the following lemma:
\begin{lemma}
The probability of a phase to be a bad one is upper bounded by $0.4$.
\end{lemma}
\begin{proof}
Let $\phi$ be a phase. Assume that $k$ verifies the conditions in
\emph{(i)}, i.e., $k\leq n'$ or $k\geq 3n'$. To prove the lemma, we
study both cases separately.
\begin{itemize}
\item $k\leq n'$, then $k$ will decrease iff no node
  beeps. This happens with probability $p = \left(1 - \frac 1
  k\right)^{n'}$, which is maximum when $k=n'$ and always less than $e^{-1}$, which is less that $0.4$.
\item $k\geq 3 n'$, then $k$ will increase iff at least two nodes beep. This is a particular case of (thus is less likely than) having at least {\em one} node beep, which happens with probability $p = 1-\left(1-\frac{1}{k}\right)^{n'}$. This formula is maximum when $k=3n'$ and it is always less than $\frac{1}{3}$, which is less than $0.4$.\qed
\end{itemize}
\end{proof}
We also have the following lemma:
\begin{lemma}
Let $\phi$ be a phase. If $k\in I$ at the start of $\phi$, then with
probability at least $1/5$, a node will be counted.
\end{lemma}
\begin{proof}
The probability that exactly one node beeps (and thus is
counted) in phase $\phi$ is given by 
$p = \frac {n'} k\left(1-\frac 1
k\right)^{n'-1}.$ This is a decreasing function of $k$ for $k>n'$ and
thus it is lower bounded when taking $k=3n'$. This gives a lower bound
of $0.2388$, which ends the proof.
\qed
\end{proof}
We now give the main result of this section:
\begin{theorem}
Let $G$ be a one-hop network of size $n$. The execution of Algorithm~\ref{algo:size} in $G$ ends in at most $55n$ phases ($165n$ slots), with
high probability.
\end{theorem}
\begin{proof}
We denote by $T$ the number of phases, $T_I$ the
number of phases where $k$ is in the interval $I$, and by $B_T$ the
number of bad phases in $T$. We have the following fact:
\begin{fact}
\label{fact:TI}
\hspace{-5pt}{\bf .}\hspace{5pt}
$T_I \geq T - 4n - 2 \times B_T$, where $4n$ is an upper bound on the number of ``external'' phases, that is, $n-2$ initial phases from $k=2$ to $k=n$ and $3n-3$ phases where $k$ decreases from $3n$ to $3$. The factor $2$ for $B_T$ corresponds to the fact that bad phases have symmetrical ``good phases'' where $k$ comes closer to the good interval (and their number is at most $B_T$).
\end{fact}

Note that $B_T$ is dominated by a binomial random variable $BIN(T,0.4)$
with parameters $T$ and $0.4$. Then, taking $T=55n$, one can use the
Chernoff bound to obtain:
\begin{eqnarray}
{\mathbb P}r\left( B_{T} \geq 23n\right) & = & {\mathbb P}r\left( B_{T} - 0.4\times T \geq n\right)\\\nonumber
& \leq & 2\times e^{-n^2/3\times0.4\times T} = 2\times e^{-n/66}.\nonumber
\end{eqnarray}
This shows that if $T=55n$ then, w.h.p., the number of bad phases is
upper bounded by $23n$.  Thus, using Fact~\ref{fact:TI}, we obtain
that, w.h.p., the number of phases inside interval $I$ is
at least $5n$. Since each phase inside $I$ counts a node with
probability at least $1/5$, we can use the same arguments to prove
that, w.h.p., the number of nodes counted during $T=55n$ phases is
at least $n$. Which ends the proof.\qed
\end{proof}
From Lemma~\ref{lowerbound}:
\begin{corollary}
Algorithm~\ref{algo:size} is optimal up to a
constant factor.
\end{corollary}

\subsubsection{Experimental Results:}
The algorithm was implemented and run on complete graphs with size $n$ ranging over all powers of $2$ from $8$ to $512$ (inclusive). For each value of $n$, we performed $10000$ runs and measured the number of phases before the algorithm terminates.\footnote{Source code available upon request.} A linear regression with Gnuplot gives us an average number of phases of $3.3197n$ (with very small regression error of $\pm 0.03074\%$). Since there are $3$ slots per phases, the expected slot complexity in practice is less than $10n$.

\section{Adaptation of the algorithm in other beeping models}
\label{sec:adaptation}

This section explores the possible adaptations of Algorithm~\ref{algo:size} in other variants of the beeping model, namely in $B_{cd}L_{cd}$ (the strongest), in $BL$ (the weakest), and in $BL_{cd}$. We show that the extra power available in $B_{cd}L_{cd}$ allows us to reduce the slot complexity by a constant factor. Adaptation to $BL$ comes at the price of sacrifying certainty (the adaptation is a Monte Carlo algorithm), it consists in emulating $B_{cd}$ beeps using techniques from~\cite{CMRZ16}. Finally, we briefly discuss the case of $BL_{cd}$, in which naive adaptations do not seem to work.

\subsection{Adaptation of the algorithm in $B_{cd}L_{cd}$}
In the $B_{cd}L_{cd}$ variant, both beeping and listening nodes can detect if several beeps occur simultaneously. In our algorithm, the purpose of the second slot is for a node to inform others that it was the only one to beep. This step is no longer necessary since listening nodes can detect it from the first slot. As a result, the phases of the algorithm can be simplified from $3$ to $2$ slots and the expected running time of the algorithm decreases by one third. (According to our experimentations, this would thus decrease the expected number of slots from less than $10n$ to less than $6.67n$.)

\subsection{Adaptation of the algorithm in $BL$}
\label{emuler}

\newcommand{\emulate}{\ensuremath{{\tt EmulateB_{cd}inBL}}\xspace}

As discussed in Section~\ref{sec:definitions}, no algorithm can solve
the counting problem with guaranteed result in $BL$, whether it be deterministic or Las Vegas. However, we can adapt it into a Monte Carlo algorithm
(uncertain result) using emulation techniques from~\cite{CMRZ16}. We
describe this adaptation here and analyse the resulting uncertainty.
We show that this uncertainty can be bounded by any constant threshold
if the nodes know an upper bound on the size of the network. Otherwise, we must settle for a certainty bound that depends on the size of the network (which is unknown).

The adaptation consists in replacing the {\tt beep} instruction in
Line~\ref{line:beep} of Algorithm~\ref{algo:size} by a call to the
procedure \emulate. The other {\tt beep} instructions in
Algorithm~\ref{algo:size} need not be replaced. Procedure \emulate
emulates a $B_{cd}$ beep (beep with collision detection) by means of
several basic $B$ beeps. It is itself organised into phases of
two slots. In each of these phases, the node beeps in one of the two
slots, chosen uniformly at random, and it listens in the other. The
intuitive idea is that if several nodes are to execute this procedure
simultaneously, then they will eventually beep in different slots and
thus detect it. Unfortunately, so long as the number of phases is
finite, there is a non-zero probability that these nodes always
beep simultaneously in the procedure, which is why the resulting
algorithm is Monte Carlo.

The detail of the procedure is given in Algorithm~\ref{algo:emulateBcd}. 
For technical reasons, we have the nodes generate their random bits (called signature) beforehand and use the same every time the emulation procedure is called. Besides using less random bits, this improves the probability of success by avoiding the extra union bound due to repeating a successful emulation several times. Let $r$ be the number of round in each emulation (it value is discussed later on), we denote by $s$ the signature of a node, that is, the word formed by $r$ bits generated uniformly at random.


\begin{algorithm}[!h]
    \label{algo:emulateBcd}
    \caption{A Procedure to emulate a $B_{cd}$  in  the $BL$ model.}  
    \textbf{Procedure }$EmulateB_{cd}inBL$(\textbf{IN:}{$s:$ word of bits associated to the vertex}\textbf{;} \textbf{OUT: }{$collision: boolean$})\\
    $Boolean\ collision \gets  false;$
    
    $Integer\ i \gets  0;$
    -
    \Repeat{$i =  r$ }{
      \Switch{slot}{ 
        \uCase{1}{ \lIf{$s[i]=0$}{beep}\lElse{listen}
        }
        \uCase{2}{\lIf{$s[i]=1$}{beep}\lElse{listen} 
        } 
        \SetKwSwitch{Switch}{Case}{}{}{}{end of phase:}{}{}{}%
        \uCase{}{
          \If{someone beeped in a different slot than I did}{$collision \gets  true$}
        }
      }
      $i \gets  i+1$
    }
    
    \textbf{End Procedure}
 \end{algorithm}


    
    
           
      
      
      
    

\subsubsection{Analysis of the Algorithm:}

We say that the procedure is {\em correct} if it detects when several nodes execute it simultaneously. The probability that it fails is maximum when two nodes only execute it, which corresponds to $1/2^r$. Hence, the probability that the procedure is not correct is upper bounded by $1/2^r$. The results in Lemma~\ref{lem:emule} are obtained by replacing $r$ with different values.

\begin{lemma}\label{lem:emule}
For any $\varepsilon>0$ and for any $n>0$:
\begin{enumerate} 
\item if $r=\lceil \log\left(\frac n\varepsilon\right) \rceil$, then the procedure is correct with probability at least $1-\frac{\varepsilon}{n} > 1-\varepsilon$, 
\item if $r=\lceil \log\left(\frac 1 \varepsilon\right) \rceil$, then, for any node $v$, the procedure
is correct on $v$ 
with probability at least $1-\varepsilon$, 
\item if $r = \lceil 2\log(n)\rceil$,  then, the procedure is correct w.h.p.
\end{enumerate}
\end{lemma}

Since $n$ is not known a priori, only the second result can be used effectively by the algorithm.
From Lemma~\ref{lem:emule} we obtain:
\begin{theorem}
For any graph $G$ of size $n$ and any $0<\varepsilon<1$:
 If $r=\lceil \log\left(\frac 1 \varepsilon\right) \rceil$,
  each node $v$ computes the size of the network in 
$O\left(n\log(\frac1{\varepsilon})\right)$ slots, 
and the result is correct with probability at least  $1-\varepsilon$.
\end{theorem}

Now, if the nodes know an upper bound $N$ on the size of the network, then the other two results from Lemma~\ref{lem:emule} can be used effectively by the algorithm as follows.
\begin{proposition}
For any graph $G$ of size $n$ and any $0<\varepsilon<1$:
\begin{itemize}
\item 
if $r=\lceil \log\left(\frac N \varepsilon\right) \rceil$,
the running time of the main algorithm is 
$O\left(n\log(\frac{n}{\varepsilon})\right)$, 
and the result is correct with probability at least  $1-\varepsilon$.
\item 
 if $r = \lceil 2\log(N) \rceil$,
the running time of the main algorithm is
$O\left(n\log n\right)$, and the result is correct 
with probability $1-o\left(\frac 1 n \right)$.
\end{itemize}
\end{proposition}

The latter is particularly relevant, as it implies that the correct result can be obtained with high probability. Furthermore, assuming that an upper bound is known is reasonable in practical scenarios. For instance, when a sensor network is deployed, one might want to learn afterhand how many sensors (among the initial amount) successfully started to operate. Similarly, when sensors use sleeping patterns (some are active, others sleep to save energy), this makes it possible to learn the number of active ones.

\subsection{Adaptation of the algorithm in $BL_{cd}$}

Now that we have seen how to adapt the algorithm in the strongest ($B_{cd}L_{cd}$) and the weakest ($BL$) beeping models, we turn our attention to the symmetrical $BL_{cd}$ model. In this variant, only listening nodes can detect simultaneous beeps. Algorithm~\ref{algo:size} can be adapted into this model. The basic idea is that instead of having a node inform the others when it detects that it beeped alone, we could do the reverse and have listeners inform the beeping node when it beeped alone. Hence, the first and last slot of each phase remain as in Algorithm~\ref{algo:size}. The second slot is split into two slots. In the first of these (new slot $2$), any node that was listening in slot $1$ beeps. This is to detect pathetic cases where {\em all} the nodes beep at the same time. In the second of these slots (new slot $3$), the listening node(s) beep if they detected a collision in slot~$1$. Based on this beep, if a node beeped in slot $1$, it detects that it was alone (other nodes also know it, since they were all listening in slot $1$). The only exception, which new slot $2$ can't avoid, is the special case that $n=1$ from the start. In this case, the node cannot decide whether it is alone to beep (and thus must count itself) or several nodes are beeping. This problem does not exist if $n \ge 2$, even when a single node remains to be counted in the end (since counted nodes keep on listening and participating until termination).

As for the slot complexity, the number of phases still depends solely on the drawings of slot $1$ and thus remains unchanged. The fact that every phase is now composed of $4$ slots makes the running time increase by one third.  (According to our experimentations, this would thus increase the expected number of slots from less than $10n$ to less than $13.34n$.)

\section{Concluding remarks}

In this paper, we presented an algorithm for counting the number of processes in one-hop beeping networks. Due to natural liminations of the beeping model, we observed that no algorithm can solve the problem with guaranteed result unless some form of collision detection is available (either for beeping nodes or for listening nodes). We presented an algorithm in the case that beeping nodes can detect collision ($B_{cd}L$), and prove that it runs in linear time in the number of nodes. We proved that this is optimal up to a constant factor and estimated the expected value of this factor experimentally. Interestingly, the algorithm can be adapted in $B_{cd}L_{cd}$ with a small gain in slot complexity ($1/3$ less slots) and in $BL_{cd}$ with a small overhead ($4/3$ more slots). Its adaptation in $BL$ is not as direct, but we proved that the algorithm can be emulated in this variant by a Monte Carlo algorithm which runs in a larger, though reasonable amount of time ($\log n$ more slots). Whether this is optimal for $BL$ is an open question.

\bibliographystyle{plain} 
\bibliography{biblio}
\end{document}